\documentclass[12pt]{article}
\usepackage[margin=1in]{geometry} 
\usepackage{amsmath,amsthm,amssymb,graphicx,mathtools,tikz,hyperref}
\usetikzlibrary{positioning}
\usepackage{hyperref}
\usepackage{pgf,tikz}
\usepackage{tkz-euclide}
\usepackage{subfig}
\usetikzlibrary{arrows}
\usepackage{hyperref}
\usetikzlibrary{automata,positioning}
\usetikzlibrary{decorations.pathmorphing,shapes}
\usepackage[ruled,vlined,linesnumbered]{algorithm2e}

\newcommand{\norm}[1]{\left\lVert#1\right\rVert}
\newcommand{\V}{\mathbf{v}}
\newcommand{\U}{\mathbf{u}}
\newcommand{\p}{\mathbf{p}}
\newcommand{\q}{\mathbf{q}}
\newcommand{\One}{\mathbf{1}}
\newcommand{\M}{\mathbf{M}}
\newcommand{\A}{\mathbf{A}}
\newcommand{\Lp}{\mathbf{L}}
\newcommand{\e}{\mathbf{e}}

\newtheorem{theorem}{Theorem}[section]
\newtheorem{proposition}[theorem]{Proposition}

\newtheorem{lemma}[theorem]{Lemma}

\newtheorem{definition}[theorem]{Definition}

\usepackage[textsize=scriptsize,textwidth=3cm,backgroundcolor=orange!20]{todonotes}

 \hypersetup{
 colorlinks,
 linkcolor=blue,
 citecolor=green
 }
\begin{document}
\SetEndCharOfAlgoLine{}
\date{}

\title{Spectral methods for testing cluster structure of graphs}
\author{Sandeep Silwal\thanks{Massachusetts Institute of Technology, Cambridge, MA 02139. Email:~\url{silwal@mit.edu}.  Research supported by the Raymond Stevens Fund.}
\and  Jonathan Tidor\thanks{Massachusetts Institute of Technology, Cambridge, MA 02139. Email:~\url{jtidor@mit.edu}. Research partially supported by an NSF Graduate Research Fellowship.}}

\maketitle
\begin{abstract}
 In the framework of graph property testing, we study the problem of determining if a graph admits a cluster structure. We say that a graph is $(k, \phi)$-clusterable if it can be partitioned into at most $k$ parts such that each part has conductance at least $\phi$. We present an algorithm that accepts all graphs that are $(2, \phi)$-clusterable with probability at least $\frac{2}3$ and rejects all graphs that are $\epsilon$-far from $(2, \phi^*)$-clusterable for $\phi^* \le  \mu \phi^2 \epsilon^2$ with probability at least $\frac{2}3$ where $\mu > 0$ is a parameter that affects the query complexity. This improves upon the work of Czumaj, Peng, and Sohler by removing a $\log n$ factor from the denominator of the bound on $\phi^*$ for the case of $k=2$. Our work was concurrent with the work of Chiplunkar et al.\@ who achieved the same improvement for all values of $k$. Our approach for the case $k=2$ relies on the geometric structure of the eigenvectors of the graph Laplacian and results in an algorithm with query complexity $O(n^{1/2+O(1)\mu} \cdot \text{poly}(1/\epsilon, 1/\phi,\log n))$.
\end{abstract}
\section{Introduction}

In this paper we study property testing of graphs in the bounded degree model. The input is a graph $G = (V,E)$ on $n$ vertices where all the vertices have degree at most $d$. Given a graph property $\mathcal{P}$, we say that $G$ is $\epsilon$-far from satisfying $\mathcal{P}$ if $\epsilon d n$ edges need to be added or removed from $G$ to satisfy $\mathcal{P}$. A property testing algorithm for $\mathcal{P}$ is an algorithm that accepts every graph $G$ satisfying $\mathcal{P}$ with probability at least $\frac{2}3$ and rejects every graph that is $\epsilon$-far from satisfying $\mathcal{P}$ with probability at least $\frac{2}3$. 

$G$ is represented as an oracle that returns the $i$th neighbor of any vertex $v$ for any $1 \le i \le d$. If $i$ is larger than the degree of $v$, a special symbol is returned. The goal of property testing is to find an algorithm with an efficient query complexity, defined as the number of oracle queries that the algorithm performs. This framework of property testing of graphs was developed by Goldreich and Ron \cite{ron_goldreich_original} and
has been applied to study various properties such as bipartiteness \cite{bipartiteness} and 3-colorability \cite{ron_goldreich_original}. See \cite{ron_goldreich_original} and \cite{property_testing_survey1} for more examples.

Our paper deals with a generalization of property testing. We are interested in testing for a family of properties $\mathcal{P}$ that depends on a single parameter $\alpha$ and is nested, satisfying $\mathcal{P}(\alpha) \subseteq \mathcal{P}(\alpha')$ for all $\alpha \ge \alpha'$. Our goal is an algorithm which accepts graphs satisfying $\mathcal{P}(\alpha)$ with probability at least $\frac{2}3$ and rejects graphs that are $\epsilon$-far from satisfying $\mathcal{P}(\alpha')$ with probability at least $\frac{2}3$ where $\alpha \ge \alpha'$. A diagram for this generalization of property testing is shown in Figure \ref{fig:testing_order}.

\begin{figure}[!htbp] 
    \centering
    \includegraphics[scale=0.45]{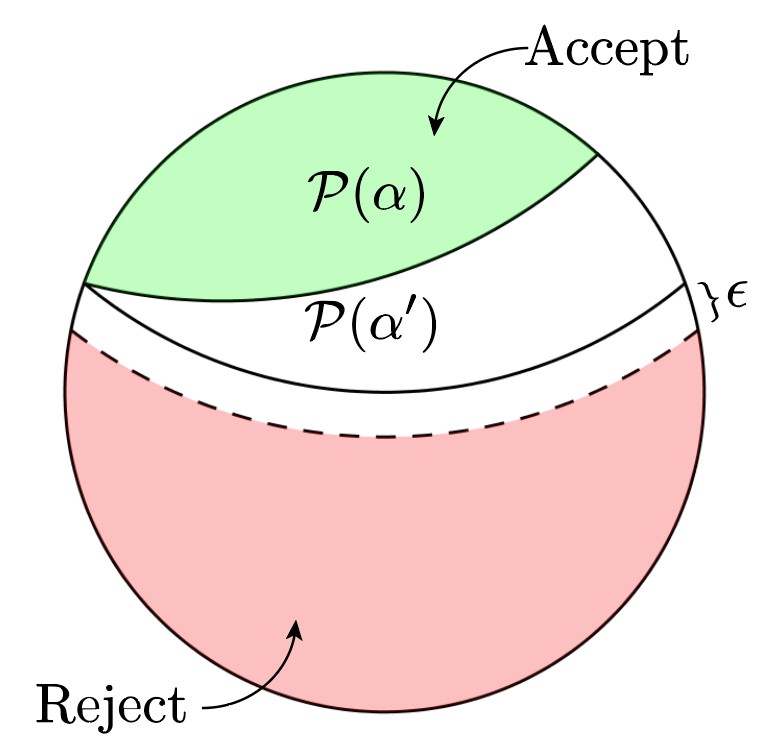}
    \caption{We want to accept graphs satisfying $\mathcal{P}(\alpha)$ and reject graphs that are $\epsilon$-far from satisfying $\mathcal{P}(\alpha')$ where $\alpha \ge \alpha'$.}
    \label{fig:testing_order}
\end{figure}

We are interested in the property of $k$-clusterability as defined by Czumaj, Peng, and Sohler in \cite{k_clusterability}. Roughly speaking, a graph is $k$-clusterable if it can be partitioned into at most $k$ clusters where vertices in the same cluster are ``well-connected.'' The connectedness of the clusters is measured in terms of their inner conductance, defined below. The idea of using conductance for graph clustering has been studied in numerous works, such as \cite{clustering_example}.

Testing for $k$-clusterability is inspired by expansion testing which has been studied extensively. A graph is called an $\alpha$-expander if every $S \subset V$ of size at most $|V|/2$ has neighborhood of size at least $\alpha |S|$. Czumaj and Sohler \cite{czumaj_sohler_expansion} showed that an algorithm proposed by Goldreich and Ron in \cite{ron_goldreich_expansion} can distinguish between $\alpha$-expanders and graphs which are $\epsilon$-far from having expansion at least $\Omega(\alpha^2/\log n)$ in the bounded degree model. This work was subsequently improved by Kale and Seshadhri \cite{kale_seshadhri} and then by Nachmias and Shapira \cite{asaf_asaf} who showed that the same algorithm distinguishes graphs which are $\alpha$-expanders from graphs which are $\epsilon$-far from $\Omega(\alpha^2)$-expanders. The work of Nachmias and Shapira also shows that expansion testing is related to the second eigenvalue of the Laplacian matrix. In addition, as shown by \cite{scooped}, testing for $k$-clusterability is related to the $(k+1)$st eigenvalue of the Laplacian so $k$-clusterability testing is a natural extension of expansion testing.

We now define conductance which is closely related to expansion. Let $S \subset V$ such that $|S| \le |V|/2$. The conductance of $S$ is defined to be $\phi_G(S) = \frac{e(S, V \setminus S)}{d|S|}$ where $e(S, V \setminus S)$ is the number of edges between $S$ and $V \setminus S$. The conductance of $G$ is defined to be the minimum conductance over all subsets $|S| \le |V|/2$ and is denoted $\phi(G)$. Now for any $S \subseteq V$, let $G[S]$ denote the induced subgraph of $G$ on the vertex set defined by $S$. We let $\phi(G[S])$ denote the conductance of this subgraph. To avoid confusion we call $\phi(G[S])$ the $\textit{inner conductance}$.

We say $G$ is $(k, \phi)$-clusterable if there exits a partition of $V$ into at most $k$ subsets $C_i$ such that $\phi(G[C_i]) \ge \phi$ for all $i$. This definition is slightly different from the one used by Czumaj, Peng, and Sohler in \cite{k_clusterability} because their definition also requires $\phi_G(C_i)$ to be bounded by $O(\epsilon^4 \phi^2)$ for all $i$. The algorithm of Czumaj, Peng, and Sohler accepts all $(k, \phi)$-clusterable graphs with probability at least $\frac{2}3$ and rejects all graphs that are $\epsilon$-far from $(k, \phi^*)$-clusterable where $\phi^* = c'_{d,k} \, \frac{\phi^2 \epsilon^4}{\log n} $ and where $c'_{d,k}$ depends only on $d,k$ \cite{k_clusterability}. Our work improves upon this result by removing the $\log n$ dependency for the case of $k=2$. 

Our main result is an algorithm in the bounded degree model that accepts every $(2, \phi)$-clusterable graph with probability at least $\frac{2}3$ and rejects every graph that is $\epsilon$-far from $(2, \phi^*)$-clusterable with probability at least $\frac{2}3$ if $\phi^* \le \mu \phi^2 \epsilon^2$ where $\mu > 0$ is a parameter that we can choose which affects the query complexity. Our algorithm has query complexity $O(n^{1/2+O(1)\mu} \cdot \text{poly}(1/\epsilon, 1/\phi,\log n))$ where $\text{poly}(1/\epsilon, 1/\phi, \log n)$ denotes a polynomial in $1/\epsilon, 1/\phi$,  and $\log n$.

The work of Czumaj et al.\@ for testing $k$-clusterability uses property testing of distributions, such as testing the $l_2$ norm of a discrete distribution and testing the closeness of two discrete distributions. For some work on testing the norm of a discrete distribution and testing closeness of discrete distributions, see \cite{czumaj_sohler_expansion} and \cite{l2_closeness_tester} respectively.

Our work was concurrent with the work of Chiplunkar et al.\@ \cite{scooped} who give an algorithm that for any fixed $k$ accepts $(k, \phi)$-clusterable from graph with probability at least $\frac{2}3$ and rejects every graph that is $\epsilon$-far from $(k, \gamma\phi^2)$-clusterable with probability at least $\frac{2}3$ using $O(n^{1/2+O(\gamma)})$ queries. This matches the query bound achieved by Nachmias and Shapira in the expander testing setting ($k=1$). The algorithm of Chiplunkar et al.\@ also looks at the $(k+1)$st largest eigenvalue of a transformation of the lazy random walk matrix $\M$, and accepts if this eigenvalue is below a certain threshold. We essentially employ the same approach in the case of $k=2$. Our proof of correctness is rather different, relying on the geometric properties of the endpoint distributions of random walks on the input graph to deduce the size of the eigenvalues of $\M$. 

We present our algorithm, \textbf{Cluster-Test}, in Section \ref{sec:overview}. We prove that \textbf{Cluster-Test} accepts $(2, \phi)$-clusterable graphs in Section \ref{sec:completeness} and that it rejects graphs that are $\epsilon$-far from $(2, \mu \phi^2\epsilon^2)$-clusterable in Section \ref{sec:soundness}.

\subsection{Definitions} \label{sec:overview}
\begin{definition}[Graph clusterability] \normalfont
$G = (V, E)$ is $(2, \phi)$-clusterable if the conductance of $G$ is at least $\phi$ or $V$ can be partitioned into two subsets $C_1$ and $C_2$ such that the inner conductance of $C_i$ is at least $\phi$ for each $i \in \{1,2\}.$
\end{definition}

The motivating idea in designing \textbf{Cluster-Test} is to compute the rank of $\M - \frac{1}n\mathbf{J}$ where $\M$ is the lazy random walk matrix and $\mathbf{J}$ is the matrix of all $1$'s. Essentially, we show when $G$ is $(2, \phi)$-clusterable, $\M - \frac{1}n \mathbf{J}$ is ``close'' to a rank $1$ matrix while when $G$ is $\epsilon$-far from $(2, \mu \phi^2\epsilon^2)$-clusterable, $\M$ is not ``close'' to rank $1$. The intuition for this comes from Lemma \ref{lem:eigenvalue_gap} which tells us that the third eigenvalue of $\M$ is small if $G$ is $(2,\phi)$-clusterable.

Because computing the eigenvalues of $\M$ is too expensive, we instead look at the eigenvalues of $2$ by $2$ principal submatrices of $(\M - \frac{1}n\mathbf{J})^{2t}$. These principal submatrices are the Gram matrices of the endpoint distribution vectors of random walks on $G$ minus $\frac{1}n\One$. This allows us to show that if $G$ is $(2, \phi)$-clusterable then we can expect all of these submatrices to have at least one small eigenvalues while if $G$ is $\epsilon$-far from $(2, \phi)$-clusterable, both of the eigenvalues of most of these principal submatrices are large. This is essentially what our algorithm tests for.

Before we present our algorithm we introduce some standard definitions and tools that we use. Given a graph $G$ with maximum degree $d$, we work with the lazy random walk matrix $\M$ defined as follows: the off diagonal entries of $\M$ are $\frac{1}{2d}$ times the corresponding entry in the adjacency matrix while the diagonal entries of $\M$ are set so that the columns of $\M$ add to $1$ which corresponds to adding self loops of the appropriate weights in $G$. We then define the Laplacian matrix $\Lp$ as $2\textbf{I}-2\M$. Our definition of the Laplacian follows the convention used in \cite{k_clusterability} so that we can easily use some of their results.

Let $0 = \lambda_1 \le \lambda_2 \le \cdots \le \lambda_n \le 2$ denote the eigenvalues of $\Lp$ and let $\V_1, \ldots, \V_n$ denote the corresponding orthonormal eigenvectors. Let $\nu_1 \ge \nu_2 \ge \cdots \ge \nu_n$ denote the eigenvalues of $\M$ where $\nu_i = 1 - \frac{\lambda_i}2$ for $1 \le i \le n$. For $u \in V$, we define $\p_u^t$ to be the probability distribution of the endpoint of a length $t$ lazy random walk that starts at vertex $u$. That is, 
\begin{equation} \label{eq:p_expansion}
    \p_u^t = \M^t\One_u  = \sum_{i=1}^n \V_i(u) \left (1-\frac{\lambda_i}2 \right)^t\V_i,
\end{equation}
where $\One_u$ is the vector with $1$ in the entry corresponding to the vertex $u$ and $0$ elsewhere. Because $\V_1=\frac{1}{\sqrt{n}}\One$, we typically work with  $\q_u^t = \p_u^t - \frac{1}n \One$ for convenience. From Eq.\@ \eqref{eq:p_expansion}, we have
\begin{equation} \label{eq:q_expansion}
   \q_u^t = \p_u^t - \frac{1}n \One = \sum_{i \ge 2}\V_i(u)\left(1 - \frac{\lambda_i}{2} \right)^t \V_i.
\end{equation}

\subsection{Preliminary Results}
In this paper $\norm{\cdot}$ always denotes the $l_2$ norm unless stated otherwise. We need the following classical result from \cite{weyl} which roughly states that eigenvalues are stable under small perturbations.

\begin{proposition}[Weyl's Inequality] \label{thm:weyl} Let $\mathbf{B} = \widetilde{\mathbf{B}}+\mathbf{E}$ and suppose $\mathbf{B}$ has eigenvalues $\mu_1 \ge \cdots \ge \mu_n$ and $\widetilde{\mathbf{B}}$ has eigenvalues $\tilde{\mu}_1 \ge \cdots \ge \tilde{\mu}_n$. Furthermore, suppose $\norm{\mathbf{E}}_F \le \epsilon$ where $\norm{ \cdot }_F$ denotes the Frobenius norm.  Then $|\mu_i - \tilde{\mu}_i| \le \epsilon$ for all $1 \le i \le n$.
\end{proposition}
Our work relies on estimating dot products and norms of various distributions where we view distributions over $n$ elements as vectors in $\mathbb{R}^n$. To estimate these quantitites, we use the following result about distribution property testing.

\begin{theorem}[Theorem 1.2 in \cite{l2_closeness_tester}] \label{thm:dotproduct}
Let $\p,\q$ be two distributions with $b \ge \max(\norm{\p}^2, \norm{\q}^2)$. There is an algorithm $l_2 \textbf{-Inner-Product-Estimator} (\eta, \xi, b, \p, \q)$ which computes an estimate of $\langle \p, \q \rangle $ that is accurate to within additive error $\xi$ with probability at least $1-\eta$ and requires $c_{2.2} \, \frac{\sqrt{b}}{\xi} \log \frac{1}{\eta}$ samples from each of the distributions $\p$ and $\q$ for some absolute constant $c_{2.2}$.
\end{theorem}

\begin{theorem}[Lemma 3.2 in \cite{k_clusterability}] \label{thm:l2norm}
Let $\p \in \mathbb{R}^n$ be a distribution over $n$ elements. There is an algorithm $l_2 \textbf{-Norm-Tester}(\sigma, r ,\p)$ that accepts if $\norm{\p}^2 \le \frac{\sigma}4$ and rejects if $\norm{\p}^2 \ge \sigma$ with probability at least $1-\frac{16 \sqrt{n}}{r}$ and requires $r$ samples from $\p$. A condition on the input $r$ is that it must be at least $16 \sqrt{n}$.
\end{theorem}

\section{Algorithm}
We now describe our algorithm \textbf{Cluster-Test}. Our algorithm performs multiple lazy random walks on the input graph and uses the distribution testing results from Theorems \ref{thm:dotproduct} and \ref{thm:l2norm} to approximate a principal submatrix of $(\M - \frac{1}n\mathbf{J})^{2t}$. As shown in Sections \ref{sec:completeness} and \ref{sec:soundness}, our choice of $t$ in Theorem \ref{thm:main} is large enough so that a random walk of length $t$ mixes well in the case that $G$ is $(2, \phi)$-clusterable and small enough so that the random walk does not mix well if $G$ is $\epsilon$-far from $(2, \mu\phi^2\epsilon^2)$-clusterable. 

We use the notation $\A_{u,v}$ for the $2$ by $2$ submatrix of $(\M - \frac{1}n\mathbf{J})^{2t}$ with rows and columns indexed by the vertices $u$ and $v$. Noting that $(\M - \frac{1}n \mathbf{J})^{2t} = \M^{2t} - \frac{1}n\mathbf{J}$, we can write
 \[ \A_{u,v} = 
\begin{bmatrix}
       \norm{\q_u^t}^2    & \langle \q_u^t, \q_v^t \rangle \\
       \langle \q_v^t, \q_u^t \rangle & \norm{\q_v^t}^2
\end{bmatrix}.
\]
Note that we assume $\A_{u,v}$ depends on the parameter $t$ that is inputted into $\textbf{Cluster-Test}.$ 

\begin{algorithm}
    \For{$R$ rounds}
    {Pick a pair of vertices $u$ and $v$ uniformly at random from $G$. \; 
    Run $N$ random walks of length $t$ starting from $u$ and starting from $v$.  \;
    Compute $l_2\textbf{-Norm-Tester}(\sigma, r, \p_u^t)$ and $l_2\textbf{-Norm-Tester}(\sigma, r, \p_v^t)$ using the $N$ samples of $\p_u^t$ and $\p_v^t$ from step 3. If either trial rejects, abort and reject $G$. \;
    Compute $l_2 \textbf{-Inner-Product-Estimator}(\eta, \xi, \sigma/4, \p_u^t, \p_v^t)$ with the results of step $3$ to approximate each entry of $\A_{u,v}$ within additive error $\xi$ for each entry. Call the approximation $\tilde{\A}_{u,v}$. \;
    Abort and reject $G$ if both the eigenvalues of $\tilde{\A}_{u,v}$ are larger than $\Lambda$. \;
    }
    Accept $G$.
    \caption{$\textbf{Cluster-Test}(G, R, t, \eta, \sigma, \xi, N, r, \Lambda)$}
    \label{alg:cluster-test}
\end{algorithm}
We now present our main theorem about the guarantees of $\textbf{Cluster-Test}$.
\begin{theorem}[Main Theorem] \label{thm:main}
Let $G$ be an $n$ vertex graph with maximum degree at most $d$. For any $\mu \in (0, C)$ we set
$$ R = \frac{10^{22}}{\epsilon^4}, \quad t =  \frac{64\max(c_{3.3},c_{3.5})\log n}{\phi^2}, \quad \eta = \frac{1}{24R}, \quad  \sigma =  \frac{16}{\eta n}, $$
$$ \xi = \frac{1}{10^5} \, \frac{1}{n^{1+128c_{3.3}c_{3.10}\mu}}, \quad  N = c_{2.2} \, \frac{\sqrt{\sigma}}{2\xi} \log \frac{1}{\eta}, \quad  r = \frac{16 \sqrt{n}}{\eta}, \quad  \Lambda = \frac{1}{10^4} \, \frac{1}{n^{1+128c_{3.3}c_{3.10}\mu}}, $$ where $C = \frac{1}{128c_{3.3}c_{3.10}}$ and the constants $c_{2.2}, c_{3.3}, c_{3.5}$ and $c_{3.10}$ are defined in Theorem \ref{thm:dotproduct} and Lemmas \ref{lem:eigenvalue_gap}, \ref{lem:small_norm}, and \ref{lem:sparse_cuts} respectively. Then,
\begin{enumerate}
\item \textbf{Cluster-Test} with the parameters defined above accepts every $(2, \phi)$-clusterable graph $G$ with probability at least $\frac{2}3$.
\item \textbf{Cluster-Test} with the parameters defined above rejects every graph $G$ that is $\epsilon$-far from $(2, \phi^*)$-clusterable for any $\phi^* \le \mu \phi^2 \epsilon^2$ with probability at least $\frac{2}3$.
\end{enumerate}
Furthermore, the query complexity of $\textbf{Cluster-Test}$ is $O(n^{1/2+O(1)\mu} \cdot \textup{poly}(1/\epsilon, 1/\phi,\log n))$.
\end{theorem}

\section{Proof of Main Theorem \label{sec:proof}} 
\subsection{Completeness: accepting \texorpdfstring{$(2, \phi)$}{Lg}-clusterable graphs} \label{sec:completeness}
In this section we show that $\textbf{Cluster-Test}$ with the parameters defined in Theorem \ref{thm:main} accepts $G$ with probability greater than $ \frac{2}3$ if $G$ is $(2, \phi)$-clusterable. We first introduce the main geometric property of our paper.
\begin{definition} \label{defn:epsilon_close}
\normalfont Vectors $\textbf{a}$ and $\textbf{b}$ are \emph{$\epsilon$-close to collinear} if they can be moved $l_2$ distance at most $\epsilon$ to lie on a line through the origin. Vectors $\textbf{a}$ and $\textbf{b}$ are \emph{$\epsilon$-far from collinear} if they are not $\epsilon$-close to collinear. See Figure \ref{fig:antipodal} for reference.
\end{definition}
Let $G$ be a $(2, \phi)$-clusterable graph. We show that \textbf{Cluster-Test} accepts $G$ with probability at least $\frac{2}3$ using the following argument.
\begin{itemize}
    \item First in Lemma \ref{lem:eigsofA} we first show that how close $\q_u^t$ and $\q_v^t$ are to collinear corresponds to how small the eigenvalues of $\A_{u,v}$ are.
    \item We show in Lemma \ref{lem:norm_difference} that any pair of vectors $\q_u^t$ and $\q_v^t$, where $u,v$ are vertices of $G$, are very close to collinear. This relies on a result about the eigenvalues of $\M$ from \cite{k_clusterability} which is restated in Lemma \ref{lem:eigenvalue_gap}.
    \item We finally show that this implies that $\textbf{Cluster-Test}$ accepts $G$ with probability greater than $\frac{2}3$ in Lemma \ref{lem:algproof_completeness}.
\end{itemize}
\begin{lemma} \label{lem:eigsofA}
If $\q_u^t$ and $\q_v^t$ are $\epsilon$-close to collinear then the smallest eigenvalue of $\A_{u,v}$ is less than $10\epsilon$. Conversely, if $\q_u^t$ and $\q_v^t$ are $\epsilon$-far from collinear then both the eigenvalues of $\A_{u,v}$ are larger than $\epsilon^2$.
\end{lemma}
\begin{proof}
Write $\M_1 = \begin{bmatrix}
       \q_u^t    &  \q_v^t 
\end{bmatrix}$,  the matrix with columns $\q_u^t$ and $\q_v^t$. Then $\A_{u,v} = \M_1^T\M_1$. Because $\A_{u,v}$ is positive semidefinite, we can also write
$$\A_{u,v} =  \kappa_1^2 \textbf{w}_1 \textbf{w}_1^T + \kappa_2^2 \textbf{w}_2 \textbf{w}_2^T$$ where  $\textbf{w}_1, \textbf{w}_2 \in \mathbb{R}^2$ are orthonormal and $\kappa_1, \kappa_2 \ge 0$.
Suppose $\q_u^t$ and $\q_v^t$ are $\epsilon$-close to collinear. An equivalent formulation of Definition \ref{defn:epsilon_close} is that there exists $\e_u, \e_v$ such that $\norm{\e_u}, \norm{\e_v} \le \epsilon$ and $\q_u^t+\e_u$ and $\q_v^t+\e_v$  lie on a line through the origin. This implies that the matrix 
$\mathbf{E} =  \begin{bmatrix}
       \e_u    &  \e_v 
\end{bmatrix} $
with columns $\e_u, \e_v \in \mathbb{R}^n$ is such that $\M_1+\mathbf{E}$ is rank $1$. Therefore, $(\M_1 + \mathbf{E})^T(\M_1 + \mathbf{E})$ is also a rank $1$ matrix so it has a zero eigenvalue. Because
$$ (\M_1 + \mathbf{E})^T(\M_1 + \mathbf{E}) = \A_{u,v} + \mathbf{E}^T\M_1 + \M_1^T\mathbf{E} + \mathbf{E}^T\mathbf{E},$$ 
we know by Weyl's inequality that $\A_{u,v}$ has an eigenvalue less than 
$$\norm{\mathbf{E}^T\M_1 + \mathbf{E}\M_1^T + \mathbf{E}^T\mathbf{E}}_F \le 2\norm{\mathbf{E}^T\M_1}_F + \norm{ \mathbf{E}^T\mathbf{E}}_F$$
where $\norm{\cdot}_F$ denotes the Frobenius norm. Because $\norm{\q_u}, \norm{\q_v} \le 1$, we can easily compute that 
\begin{align*}
    \norm{\mathbf{E}^T\M_1}_F &\le \sqrt{2}(\norm{\e_u} + \norm{\e_v}) \\
    \norm{ \mathbf{E}^T\mathbf{E}}_F &\le \norm{\e_u}^2 + \norm{\e_v}^2.
\end{align*}
Therefore, 
$$2\norm{\mathbf{E}^T\M_1}_F + \norm{ \mathbf{E}^T\mathbf{E}}_F  \le 4\sqrt{2} \epsilon + 2 \epsilon^2 \le 10 \epsilon$$
for $\epsilon \le 1$ which proves the first part of our lemma. 

For the second part, we prove the contrapositive. Suppose that $\kappa_2 < \epsilon$. We wish to show that $\q_u^t$ and $\q_v^t$ are $\epsilon$-close to collinear. Define
$$\M_2 := \kappa_1 \textbf{w}_1 \textbf{w}_1^T + \kappa_2 \textbf{w}_2\textbf{w}_2^T.
$$ 
Let $\mathbf{E} = -\kappa_2 \mathbf{w}_2 \mathbf{w}_2^T$ and denote the columns of $\mathbf{E}$ as $\e_u', \e_v' \in \mathbb{R}^2$. Then $\M_2 + \mathbf{E}$ is rank $1$ and $\norm{\e_u'}, \norm{\e_v'} \le \kappa_2 < \epsilon$. By the orthogonality of $\textbf{w}_1$ and $\textbf{w}_2$,  we have that $\M_1^T\M_1 = \M_2^T \M_2$. Therefore, the matrix $\mathbf{U} = \M_1\M_2^{-1}$ satisfies $\mathbf{U}^T\mathbf{U}=\mathbf{I}$, implying that for any vector $\mathbf{x}\in\mathbb{R}^2$ the equation $\norm{\mathbf{U}\mathbf{x}}=\norm{\mathbf{x}}$ is satisfied. We note that 
$$ \mathbf{U}(\M_2 + \mathbf{E}) = \M_1 + \mathbf{U} \mathbf{E} $$
is also a rank $1$ matrix and we define the columns of $\mathbf{U} \mathbf{E}$ to be $\e_u , \e_v \in \mathbb{R}^n$. Because $\mathbf{U}$ preserves lengths, $\norm{\e_u}, \norm{\e_v} < \kappa_2\le\epsilon$. Thus $\q_u^t, \q_v^t$ are $\epsilon$-close to collinear, as desired.
\end{proof}

We proceed to show that if the input graph is $(2, \phi)$-clusterable then for any pair of vertices $(u,v)$, $\q_u^t$ and $\q_v^t$ are $(1-O(1)\phi^2)^{t}$-close to collinear. To show this, we need the following lemma from \cite{k_clusterability} which relates the property of $(2,\phi)$-clusterable to the eigenvalues of the Laplacian matrix. 

\begin{lemma}[Lemma $5.2$ in \cite{k_clusterability}] \label{lem:eigenvalue_gap}
There exists a constant $c_{3.3}$ depending on $d$ such that for $G$ a $(2,\phi)$-clusterable graph of maximum degree at most $d$, $\lambda_i \ge \frac{\phi^2}{16c_{3.3}}$ for $i \ge 3$ where $\lambda_i$ is the $i$-th smallest eigenvalue of the Laplacian matrix of $G$.
\end{lemma}

We note here that there is a short proof that $\A_{u,v}$ has at most one large eigenvalue if $G$ is $(2, \phi)$-clusterable. Lemma \ref{lem:eigenvalue_gap} states that $\M - \frac{1}n \mathbf{J}$ has at most one large eigenvalue, hence $(\M - \frac{1}n \mathbf{J})^t = \M^t - \frac{1}n \mathbf{J}$ also has at most one large eigenvalue. Then the Cauchy interlacing theorem implies that all the minors $\A_{u,v}$ also have at most one large eigenvalue. However, we present this longer proof that uses the definition of $\epsilon$-close to highlight the similarities between the proofs of the soundness and completeness case.

We now show that given Lemma \ref{lem:eigenvalue_gap}, it follows that $\q_u^t$ is close to the line spanned by $\V_2$.

\begin{lemma} \label{lem:norm_difference} 
Let $G$ be $(2, \phi)$-clusterable. Then for any pair of vertices $(u,v)$,  $\q_u^t$ and $\q_v^t$ are $\left(1-\frac{\phi^2}{32c_{3.3}} \right)^t$-close to collinear where $c_{3.3}$ is a constant defined in Lemma $\ref{lem:eigenvalue_gap}$.
\end{lemma}
\begin{proof}
Recall that $\q_u^t = \p_u^t - \frac{1}n \One$ where $\p_u^t$ is the probability distribution of the endpoint of a length $t$ lazy random walk starting at vertex $u$. Writing $\q_u^t$ in the eigenbasis of $\M$ gives us
$$\q_u^t = \p_u^t - \frac{1}n \One = \sum_{i \ge 2}\V_i(u)\left(1 - \frac{\lambda_i}{2} \right)^t \V_i.$$ Therefore, 
\begin{align*}
    \norm{\q_u^t - \V_2}^2 &= \sum_{i=3}^n \left( 1- \frac{\lambda_i}2 \right)^{2t} \V_i(u)^2  \\
    &\le  \left(1 - \frac{\lambda_3}2 \right)^{2t} \sum_{i=3}^n \V_i(u)^2 \\
    &\le \left (1-  \frac{\phi^2 }{32c_{3.3}}\right )^{2t}.
\end{align*}
It follows that for any vertices $u$ and $v$, $\q_u^t$ and $\q_v^t$ are $\left (1-  \frac{\phi^2 }{32c_{3.3}}\right )^{t}$-close to collinear. 
\end{proof}

Lemmas \ref{lem:eigsofA} and \ref{lem:norm_difference} together guarantee that both of the eigenvalues of $\A_{u,v}$ cannot be large. We now want to show that this also holds when the \textbf{Cluster-Test} approximates $\A_{u,v}$. We need the following lemma which tells us that $l_2\textbf{-Norm-Tester}$ accepts $\p_u^t, \p_v^t$ with high probability in step $4$ of \textbf{Cluster-Test}. This lemma is just a technicality that we need for the query complexity of Theorem \ref{thm:dotproduct}.
\begin{lemma}[Lemma $4.3$ in \cite{k_clusterability}] \label{lem:small_norm}
Let $0 < \gamma < 1$. There exists a constant $c_{3.5}$ such that for $G$ a $(2, \phi)$-clusterable graph, there exists $V' \subseteq V$ with $|V'| \ge (1-\gamma)|V|$ such that for any $u \in V'$ and any $t > \frac{c_{3.5}\log n}{\phi^2}$,   the following holds:
$$ \norm{\p_u^t}^2 \le \frac{4}{\gamma n}. $$
\end{lemma}
We now prove that $\textbf{Cluster-Test}$ with the parameters defined in Theorem \ref{thm:main} passes the completeness case.
\begin{lemma} \label{lem:algproof_completeness}
\textbf{Cluster-Test} with the parameters defined in Theorem \ref{thm:main} accepts $(2, \phi)$-clusterable graphs with probability greater than $ \frac{2}3$.
\end{lemma}
\begin{proof}
Let $G$ be a $(2, \phi)$-clusterable graph. We analyze one round of \textbf{Cluster-Test} and calculate the rejection probability of one round. Note that \textbf{Cluster-Test} samples a pair of vertices $u$ and $v$ uniformly at random from $G$ at each round. There are three ways one round can reject $G$:
\begin{enumerate}
    \item One of the vertices $u$ or $v$ in the complement of $V'$ in Lemma \ref{lem:small_norm}.
    \item $l_2\textbf{-Norm-Tester}$ rejects $\p_u^t$ or $\p_v^t$ in step $4$ of \textbf{Cluster-Test}.
    \item Both of the eigenvalues of $\tilde{\A}_{u,v}$ are larger than $\Lambda$.
\end{enumerate}
Setting $\gamma = \eta = \frac{1}{24R}$ in Lemma \ref{lem:small_norm}, we see that both $u$ and $v$ lie inside $V'$ in Lemma \ref{lem:small_norm} with probability at least $(1-\eta)^2$. Therefore, the rejection probability of case $1$ is at most $1-(1-\eta)^2\leq 2\eta$.

If $u,v \in V'$ as defined in Lemma \ref{lem:small_norm}, then $\norm{\p_u^t}^2, \norm{\p_v^t}^2 \le \frac{4}{\eta n}$. Given this along with the fact that $\sigma = \frac{16}{\eta n}$, we have that $l_2\textbf{-Norm-Tester}$ accepts $G$ with probability at least $\left(1-\frac{16 \sqrt{n}}r \right)^2 = (1-\eta)^2. $ Therefore, the rejection probability of case $2$ is also at most $2\eta$.

By Lemma \ref{lem:norm_difference}, $\q_u^t$ and $\q_v^t$ are $\left( 1 - \frac{\phi^2}{32c_{3.3}} \right)^{t}$-close to collinear. Recall that $t \ge \frac{64c_{3.3}\log n}{\phi^2}$ in Theorem \ref{thm:main}. Therefore by Lemma \ref{lem:eigsofA}, $\A_{u,v}$ has at least one eigenvalue smaller than
$$ 10\left (1-  \frac{\phi^2 }{32c_{3.3}}\right )^{t} \le 10\exp\left(- \frac{\phi^2t}{32c_{3.3}}\right) \le  \frac{10}{n^2}. $$
The matrix $\tilde{\A}_{u,v}$ that $\textbf{Cluster-Test}$ computes can be written as $\tilde{\A}_{u,v} = \A_{u,v} + \mathbf{E}$ where each entry of the $2$ by $2$ matrix $\mathbf{E}$ is at most $\xi$ with probability $1-\eta$ due to $l_2 \textbf{-Inner-Product-Estimator}$. Therefore, $\norm{\mathbf{E}}_F \le 2 \xi$ with probability $(1-\eta)^4$. If this holds, then by Weyl's inequality, $\tilde{\A}_{u,v}$ has an eigenvalue at most 
$$\frac{10}{n^2} + 2\xi < \Lambda.$$ Therefore, the rejection probability of case $3$ is at most $4\eta$.

Adding up the rejection probabilities of each of the three cases tells us that one round rejects $G$ with probability at most $8 \eta$. Thus the total probability that we reject $G$ in one of the $R$ rounds is at most $8 \eta R \le \frac{1}3$, as desired. The query complexity is $O(tNR) = O(n^{1/2+O(1)\mu} \cdot \text{poly}(1/\epsilon, 1/\phi,\log n))$.
\end{proof}

\subsection{Soundness: rejecting graphs \texorpdfstring{$\epsilon$}{Lg}-far from \texorpdfstring{$(2, \phi^*)$}{Lg}-clusterable} \label{sec:soundness}
In this section we show that \textbf{Cluster-Test} rejects $G$ with probability greater than $\frac{2}3$ if $G$ is $\epsilon$-far from $(2, \phi^*)$-clusterable for $\phi^* \le \mu \phi^2 \epsilon^2$. We introduce two properties that expand on the property of $\epsilon$-close to collinear.

\begin{definition} \label{defn:antipodal} \normalfont
Vectors \textbf{a} and \textbf{b} are \emph{$\epsilon$-close to antipodal} if they can be moved distance at most $\epsilon$ to lie on a line through the origin where the origin lies between the two moved points. Vectors $\textbf{a}$ and $\textbf{b}$ are \emph{$\epsilon$-far from antipodal} if they are not $\epsilon$-close to antipodal.
\end{definition}
\begin{definition} \label{defn:podal} \normalfont
Vectors \textbf{a} and \textbf{b} are \emph{$\epsilon$-close to podal} if they can be moved distance at most $\epsilon$ to lie on a line through the origin where the origin does not lie between the two moved points. Vectors $\textbf{a}$ and $\textbf{b}$ are \emph{$\epsilon$-far from podal} if they are not $\epsilon$-close to podal.
\end{definition}
See Figure \ref{fig:antipodal} for reference. Note that vectors \textbf{a} and \textbf{b} are $\epsilon$-far from collinear if and only if they are $\epsilon$-far from both antipodal and podal. 

\begin{figure}[!htbp]
    \centering
    \subfloat[]{{\includegraphics[scale=0.49]{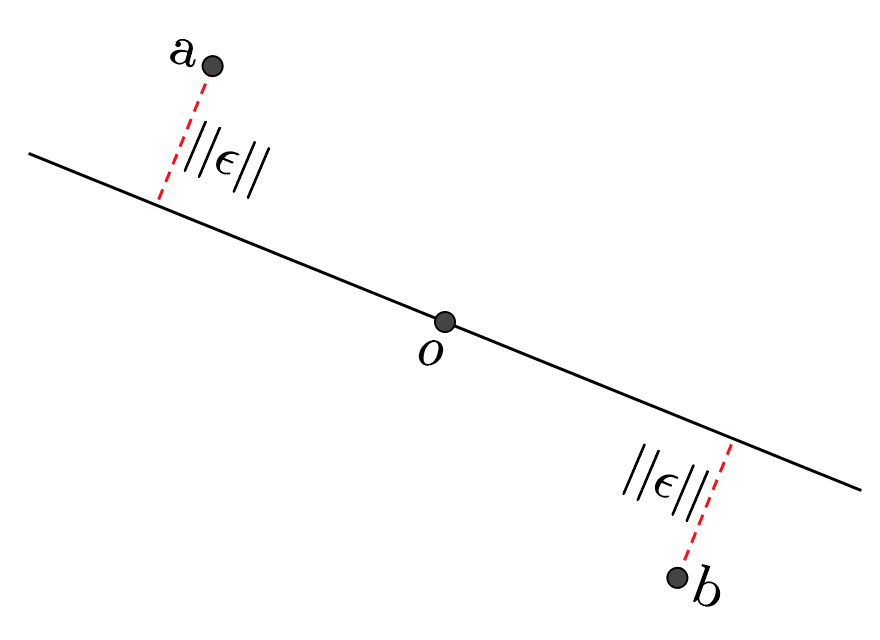} }}%
    \qquad
    \subfloat[]{{\includegraphics[scale=0.49]{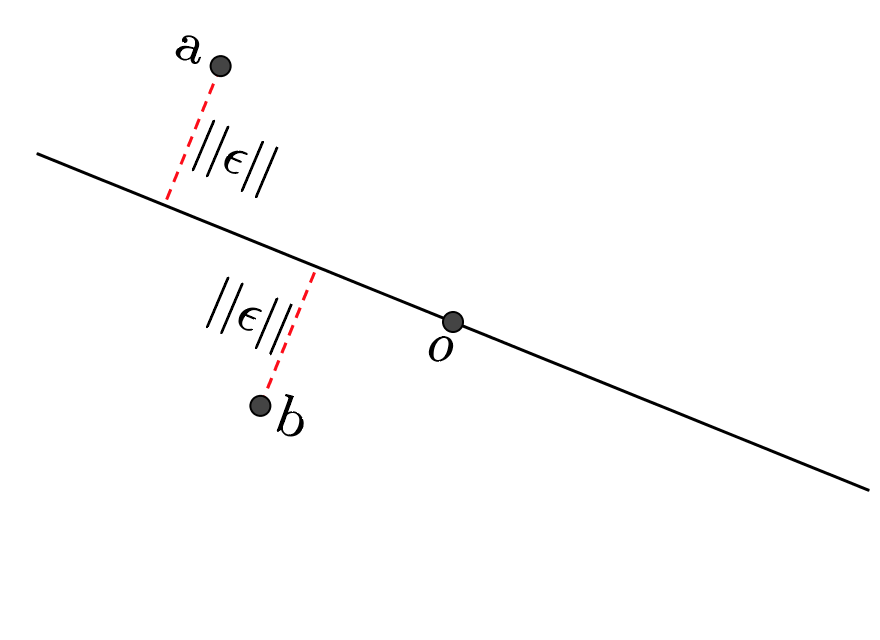} }}%
    \caption{Origin is denoted as $o$.  Vectors \textbf{a} and \textbf{b} are $\epsilon$-close to collinear in both cases. \\  (a): Vectors \textbf{a} and \textbf{b} are $\epsilon$-close to antipodal. (b): Vectors \textbf{a} and \textbf{b} are $\epsilon$-close to podal.}%
    \label{fig:antipodal}%
\end{figure}

We now outline our argument which shows that \textbf{Cluster-Test} rejects graph $G$ if $G$ is $\epsilon$-far from $(2, \phi^*)$-clusterable. We do this by showing that there are many pair of vertices $(u,v)$ where $\q_u^t$ and $\q_v^t$ are \emph{far} from collinear which allows us to say that the eigenvalues of $\A_{u,v}$ are large due to Lemma \ref{lem:eigsofA}. This is a relatively harder task than showing that $\q_u^t$ and $\q_v^t$ are close to collinear in the completeness case so we need a more complicated argument which is detailed below. For $S \subseteq V$,  we define $\q_{S}^t = \frac{1}{|S|}\sum_{u \in S} \q_u^t$.
\begin{itemize}
    \item We first present a result from \cite{k_clusterability} in Lemma \ref{lem:sparse_cuts} which says that $G$ has two large subsets of vertices $S_1$ and $S_2$ that are each separated from the rest of the vertices by sparse cuts.
    \item We let $\Pi$ be the projection onto the span of the eigenvectors of $\M$ with ``large'' eigenvalues. We use the above result to show that the aggregate vectors $\Pi \q_{S_1}^0$ and $\Pi \q_{S_2}^0$ are far from collinear in Lemma \ref{lem:large_norm}. This projection trick is necessary to relate $\q_u^t$ to $\Pi \q_u^0$ later on.
\end{itemize}
We now want to use the fact that the aggregate vectors $\Pi \q_{S_1}^0$ and $\Pi \q_{S_2}^0$  are far from collinear to find many \emph{pairs} of vectors that are far from collinear.
    
\begin{itemize}
    \item We use the pigeonhole principle to deduce that there are $\Theta(n^2)$ pairs of vertices $(u,v)$ such that $\Pi \q_u^0$ and $\Pi \q_v^0$ are far from antipodal. Similarly, we show that that there are $\Theta(n^2)$ pairs of vertices $(u,v)$ such that $\Pi \q_u^0$ and $\Pi \q_v^0$ are far from podal. This is shown in Lemma \ref{lem:many_pairs1}.
\end{itemize}
Note that the above point does \emph{not} immediately imply that there are $\Theta(n^2)$ pairs of vertices $(u,v)$ such that $\Pi \q_u^0$ and $\Pi \q_v^0$ are far from both podal and antipodal.

\begin{itemize}
    \item We use results from the previous step along with geometric properties of the vectors $\Pi \q_u^0$ to show that there are  $\Theta(n^2)$ pairs of vertices $(u,v)$ such that $\Pi \q_u^0$ and $\Pi \q_v^0$ are far from collinear in Lemmas \ref{lem:supplementary} and \ref{lem:many_pairs2}.
    
    \item Using properties of $\Pi,$ we transfer this result on the $\Pi \q_u^0$ vectors to the $\q_u^t$ vectors.
    
    \item Finally we refer back to Lemma \ref{lem:eigsofA} to argue that there are many pairs $(u,v)$ such that both the eigenvalues of $\A_{u,v}$ are sufficiently large which means that \textbf{Cluster-Test} rejects $G$ with probability at least $\frac{2}3$. This is shown in Lemmas \ref{lem:many_repr_pairs} and \ref{lem:algproof_soundness}. 
\end{itemize}
We now give quantitative versions of the definitions of antipodal and podal which is useful later on in our argument.
\begin{lemma} \label{lem:line_proj2}
If vectors $\textbf{a}$ and $\textbf{b}$ are $\epsilon$-close to antipodal then
\begin{equation} \label{eq:antipodal}
   \min_{0 \le \alpha \le 1} \norm{\alpha \textbf{a} + (1-\alpha) \textbf{b}} \le \epsilon. 
\end{equation}
Similarly, if $\textbf{a}$ and $\textbf{b}$ are $\epsilon$-close to podal then
\begin{equation} \label{eq:podal}
    \min_{0 \le \alpha \le 1} \norm{\alpha \textbf{a} + (1-\alpha) (-\textbf{b})} \le \epsilon.
\end{equation}
\end{lemma}
\begin{proof}
 If $\textbf{a}$ and $\textbf{b}$ are $\epsilon$-close to collinear then there exists $\mathbf{e}_a$ and $\mathbf{e}_b$ such that $\textbf{a} + \mathbf{e}_a$ and $\textbf{b} +  \mathbf{e}_b$ lie on the same line through the origin and $\norm{\mathbf{e}_a}, \norm{\mathbf{e}_b} < \epsilon$. If $\textbf{a}$ and $\textbf{b}$ are $\epsilon$-close to antipodal then we can find $0 \le \beta \le 1$ such that $\beta(\textbf{a} + \mathbf{e}_a) + (1-\beta)(\textbf{b} + \mathbf{e}_b) = 0$. We have
\begin{align*}
    \norm{\beta \textbf{a} + (1-\beta)\textbf{b}} &= \norm{\beta(\textbf{a} + \mathbf{e}_a) + (1-\beta)(\textbf{b} + \mathbf{e}_b) - (\beta \mathbf{e}_a + (1-\beta)\mathbf{e}_b)} \\
    &\le \norm{\beta(\textbf{a} + \mathbf{e}_a) + (1-\beta)(\textbf{b} + \mathbf{e}_b)} + \norm{\beta \mathbf{e}_a + (1-\beta)\mathbf{e}_b}  \\
    & \le \epsilon.
\end{align*}
 Therefore,
$$ \min_{0 \le \alpha \le 1}\norm{\alpha \textbf{a} + (1-\alpha)\textbf{b}} \le \epsilon$$
which proves Eq.\@ \eqref{eq:antipodal}. A similar calculation for the podal case proves Eq.\@ \eqref{eq:podal}.
\end{proof}

We restate a lemma from \cite{k_clusterability} which says that we can partition a graph that is from $(2,\phi')$-clusterable into three subsets of vertices that are separated by sparse cuts.

\begin{lemma}[Lemma $4.5$ in \cite{k_clusterability}] \label{lem:sparse_cuts}
 Let $G = (V,E)$ be a graph with maximum degree at most $d$. There are constants $\alpha$ and $c_{3.10}$, that depend on $d$, such that if $G$ is $\epsilon$-far from $(2, \phi')$-clusterable with $\phi' \le \alpha \epsilon$,  then there exists a partition of $V$ into three subsets $S_1, S_2, S_{3}$ such that for each $i \in \{1,2,3\}$,  we have $|S_i| \ge \frac{\epsilon^2|V|}{2 \cdot 10^4}$ and $\phi_G(S_i) \le c_{3.10}\phi'\epsilon^{-2}$.
\end{lemma}
From now on we assume that $S_1$ and $S_2$ are the smallest of the two parts so $$ \frac{\epsilon^2 n}{10^4} \le |S_1| + |S_2| \le \frac{2n}3$$ always holds. 

We begin by showing that a projection of the aggregate vectors $\q_{S_1}^0, \q_{S_2}^0$ are $O\left(\frac{1}{\sqrt{|S_1| + |S_2|}}\right)$-far from collinear by using tools from \cite{kale_seshadhri}. 
\begin{lemma} \label{lem:large_norm}
Let $S_1$ and $S_2$ be two disjoint subsets of vertices such that the cut $(S_i, V \setminus S_i)$ has conductance less than $\delta$ for $i \in \{1,2\}$. Suppose that $|S_1| + |S_2| \le \frac{2n}3$ and let $\Pi$ denote the projection onto the span of the eigenvectors of $\M$ with eigenvalue greater than $1-4\delta$. Then
$\Pi \q_{S_1}^0$ and $\Pi \q_{S_2}^0$ are $\frac{1}{\sqrt{|S_1|+|S_2|}}$-far from collinear.
\end{lemma}
\begin{proof}
Recall that $\Lp$ is the Laplacian and $\M$ is the lazy random walk matrix related by the equation $\Lp = 2 \mathbf{I} - 2\M$. Also recall that the eigenvalues of $\M$ are $1 = \nu_1 \ge \nu_2 \ge \cdots \ge \nu_n \ge 0$ with corresponding eigenvectors $\V_1, \ldots, \V_n$. Let $s_1 = |S_1|, s_2 = |S_2|$ and define the vector $\mathbf{f}$ as 
$$ \textbf{f}(v) = 
\begin{cases*}
      \frac{\alpha}{s_1} & if $v \in S_1$, \\
     \frac{\beta}{s_2}  & if $v \in S_2$, \\
      0 & otherwise,
\end{cases*}
$$
where $\alpha$ is any constant in $[0,1]$ and $\beta \in \{\alpha, 1-\alpha\}$. Let $\U = \textbf{f} - \frac{\alpha + \beta}{n}\One = \alpha \q_{S_1}^0 + \beta \q_{S_2}^0$. Write $\U$ in the eigenbasis of $\M$ as $\U = \sum_i c_i \V_i$. We have $\norm{\U}^2 = \sum_{i}c_i^2$ and one can compute that $\norm{\U}^2 = \frac{\alpha^2}{s_1} + \frac{\beta^2}{s_2} - \frac{(\alpha + \beta)^2}n$. Equating these two gives 
\begin{equation} \label{eq:lem311}
    \sum_{i}c_i^2 = \frac{\alpha^2}{s_1} + \frac{\beta^2}{s_2} - \frac{(\alpha + \beta)^2}n \ge \frac{\alpha^2}{s_1} + \frac{\beta^2}{s_2} - \frac{1}n.
\end{equation}
 We now also compute $\U^T \Lp \U$ in two different ways. We have $\frac{\U^T \Lp \U}2 = \norm{\U}^2 - \sum_i c_i^2 \nu_i$. On the other hand, using the quadratic form of $\Lp$ gives us $\frac{\U^T\Lp\U}{2} = \sum_{i < j} \M_{ij}(u_i - u_j)^2$.
Now note that there are three cases where the term $(u_i - u_j)^2$ is nonzero:
\begin{enumerate}
    \item One of vertex $i$ and vertex $j$ lies in $S_1$ and the other lies in $V\setminus (S_1 \cup S_2)$,
    \item One of vertex $i$ and vertex $j$ lies in $S_2$ and the other lies in $V\setminus (S_1 \cup S_2)$,
    \item One of vertex $i$ and vertex $j$ lies in $S_1$ and the other lies in $S_2$.
\end{enumerate}
In these three cases, $(u_i - u_j)^2$ evaluates to $\frac{\alpha^2}{s_1^2}$,  $\frac{\beta^2}{s_2^2}$, and $\left( \frac{\alpha}{s_1} - \frac{\beta}{s_2} \right)^2$ respectively. We bound these expressions from above by $\frac{2\alpha^2}{s_1^2}$, $\frac{2\beta^2}{s_2^2}$, and $\frac{2\alpha^2}{s_1^2} + \frac{2\beta^2}{s_2^2}$ respectively to extract the bound

$$
\sum_{i < j} \M_{ij}(u_i - u_j)^2 \le \frac{1}{2d}\left(\frac{2\alpha^2}{s_1^2}e(S_1,V\setminus S_1)+ \frac{2\beta^2}{s_2^2}e(S_2,V\setminus S_2)\right)
$$

Now using the fact that the $(S_i, V \setminus S_i)$ has conductance less than $\delta$ for each $i  \in \{1, 2\}$,  we have
\begin{align*}
    \norm{\U}^2 - \sum_i c_i^2 \nu_i &= \sum_{i < j} \M_{ij}(u_i - u_j)^2 \\
    & \le \frac{1}{2d} \frac{2\alpha^2}{s_1^2} \, \delta d s_1 +  \frac{1}{2d} \frac{2\beta^2}{s_2^2} \, \delta d s_2  \\
    &= \delta \left( \frac{\alpha^2}{s_1} +  \frac{\beta^2}{s_2}\right).
\end{align*}
It follows that 
$$ \sum_i c_i^2 \nu_i > \frac{\alpha^2}{s_1} + \frac{\beta^2}{s_2} - \frac{1}n - \delta \left( \frac{\alpha^2}{s_1} + \frac{\beta^2}{s_2} \right). $$
Call $\nu_i > 1 - 4 \delta$ ``heavy" and let $H$ be the set of indices of the heavy eigenvalues. Letting $x = \sum_{i \in H} c_i^2$,  we have 
$$x + \left( \sum_i c_i^2 -x \right)\left(1- 4 \delta \right) >  \frac{\alpha^2}{s_1} + \frac{\beta^2}{s_2} - \frac{1}n - \delta \left( \frac{\alpha^2}{s_1} + \frac{\beta^2}{s_2} \right).$$
Then using Eq.\@ \eqref{eq:lem311} implies that 
$$ x > \frac{3}4 \left( \frac{\alpha^2}{s_1} + \frac{\beta^2}{s_2} \right) - \frac{1}n.$$
By Cauchy-Schwartz,
$$ \frac{\alpha^2}{s_1} + \frac{\beta^2}{s_1} \ge \frac{|\alpha| + |\beta|}{s_1 + s_2} = \frac{1}{s_1+s_2} .$$
By assumption $\frac{2}{3(s_1+s_2)} \ge \frac{1}{n}$ so
$x \ge \frac{1}{12(s_1+s_2)}$ and hence,
\begin{equation*}
    \norm{\alpha \Pi \q_{S_1}^0  + \beta \Pi \q_{S_2}^0 }^2 \ge  \frac{1}{12(|S_1| + |S_2|)}. \qedhere
\end{equation*}
\end{proof}
We now present the following lemma which shows that the conclusions of Lemma \ref{lem:large_norm} also hold if we replace $S_1$ and $S_2$ by a large subset of themselves. This lemma is just a consequence of the triangle inequality because $\norm{\q_S^0 - \q_T^0}$ is small if $S \cap T$ is large.

\begin{lemma}\label{lem:large_norm2}
Let $S_1$ and $S_2$ be two disjoint subsets of vertices such that the cut $(S_i, V \setminus S_i)$ has conductance less than $\delta$ for $i  \in \{1, 2\}$. Suppose that $|S_1| + |S_2| \le \frac{2n}3$ and let $\Pi$ denote the projection onto the span of the eigenvectors of $\M$ with eigenvalue greater than $1-4\delta$. Let $\theta$ be a sufficiently small constant and let $T_i \subseteq S_i$ and $|T_i| \ge (1-\theta)|S_i|$ for each $i  \in \{1, 2\}$. Then $\Pi\q_{T_1}^0$ and $\Pi\q_{T_2}^0$ are $\left( \frac{1}{\sqrt{12}} - 2\sqrt{\theta} \right) \frac{1}{\sqrt{|S_1|+|S_2|}}$-far from collinear.
\end{lemma}
\begin{proof}
Let $\alpha$ be any constant in $[0,1]$ and $\beta \in \{\alpha, 1-\alpha\}$. Using the fact that $|T_i| \ge (1-\theta)|S_i|$ for each $i  \in \{1, 2\}$, we can compute that
$$ \norm{\alpha \q_{S_1}^0 + \beta \q_{S_2}^0  - \left (\alpha \q_{T_1}^0  + \beta \q_{T_2}^0 \right )}^2 \le  \frac{\theta}{1-\theta} \left( \frac{\alpha^2}{s_1} + \frac{\beta^2}{s_2} \right ) \le  2 \theta \left( \frac{\alpha^2}{s_1} + \frac{\beta^2}{s_2} \right).$$
Write $\alpha \q^t_{S_1} + \beta \q^t_{S_2}  = \sum_i c_i \V_i$ and $\alpha \q^t_{T_1} + \beta \q^t_{T_2}= \sum_i w_i \V_i$ and let $H$ denote the set of eigenvalues larger than $1- 4\delta$ as in Lemma \ref{lem:large_norm}. We have
$$ \norm{\alpha \q_{S_1}^0 + \beta \q_{S_2}^0 - \alpha \q_{T_1}^0 - \beta \q_{T_2}^0}^2 \ge \sum_{i \in H} (c_i - w_i)^2.$$ Let $S = \frac{1}{s_1 + s_2}$. From Lemma \ref{lem:large_norm} and the triangle inequality,
\begin{align*}
 \sum_{i \in H}w_i^2 &> \left( \sqrt{\sum_{i \in H}c_i^2 } - \sqrt{ \sum_{i \in H} (c_i-w_i)^2 }  \right)^2 \\ 
 &> \left( \frac{\sqrt{S}}{\sqrt{12}}  - 2\sqrt{\theta S} \right)^2 \\
  &= S\left (\frac{1}{\sqrt{12}}- 2\sqrt{\theta} \right )^2. \qedhere
\end{align*}
\end{proof}
Lemma \ref{lem:large_norm2} states that under some conditions, $\Pi \q_{T_1}^0$ and $\Pi \q_{T_2}^0$ are $O\left(\frac{1}{\sqrt{|S_1| + |S_2|}}\right)$-far from collinear. Using this result, we show that we can find many \emph{pairs} of vertices $(u,v)$ where $\Pi \q_u^0$ and $\Pi \q_v^0$ are $O\left(\frac{1}{\sqrt{|S_1| + |S_2|}}\right)$-far from antipodal but not necessarily far from podal and vice versa. 

\begin{lemma} \label{lem:many_pairs1}
Let $S_1$ and $S_2$ be two disjoint subsets of vertices such that the cut $(S_i, V \setminus S_i)$ has conductance less than $\delta$ for $i  \in \{1, 2\}$. Suppose that $|S_1| + |S_2| \le \frac{2n}3$ and let $\Pi$ denote the projection onto the span of the eigenvectors of $\M$ with eigenvalue greater than $1-4\delta$. Let $\theta$ be a sufficiently small constant. There are $\theta^2 |S_1||S_2|$ pairs $(u,v)$ where $u \in S_1, v \in S_2$,  such that $\Pi\q_u^0$ and $\Pi \q_v^0$ are
$ \sqrt{\frac{C(\theta)}{|S_1|+|S_2|}}$-far from antipodal where $C(\theta) = \frac{1}{24} - 2 \theta - \frac{\sqrt{\theta}}{\sqrt{3}}$. There are also $\theta^2 |S_1||S_2|$ pairs $(u',v')$ such that $\Pi\q_u'^0$ and $\Pi \q_v'^0$ are $\sqrt{\frac{C(\theta)}{|S_1|+|S_2|}}$-far from podal.
\end{lemma}
\begin{proof}
We first consider the antipodal case. Suppose for the sake of contradiction that there are more than $(1-\theta^2)|S_1||S_2|$ pairs of vertices $(u,v)$ such that $\Pi \q_u^0$ and $\Pi \q_v^0$ are $\sqrt{\frac{C(\theta)}{|S_1|+|S_2|}}$-close to antipodal. We now show that there is a set $T \subseteq S_1$ where $|T| \ge (1-\theta) |S_1|$ such that for all $u \in T$, there are more than $(1-\theta)|S_2|$ vertices $v$ in $S_2$ such that $\Pi \q_u^0$ and $\Pi \q_v^0$ are $\sqrt{\frac{C(\theta)}{|S_1|+|S_2|}}$-close from antipodal. This must be true because otherwise, the number of pairs that are close to antipodal is at most 
$$ \theta |S_1||S_2| + (1-\theta)^2|S_1||S_2| < (1-\theta^2)|S_1||S_2|.$$ 
 Hence, such a set $T$ must exist. Now for every $u \in T$,  let $T_u$ denote the set of vertices in $S_2$ such that $\Pi \q_u^0$ and $\Pi\q_v^0$ are $\sqrt{\frac{C(\theta)}{|S_1|+|S_2|}}$-close to antipodal for all $v \in T_u$. 
 
 We now claim that $\Pi \q_u^0$ and $\Pi \q_{T_u}^0$ are $\sqrt{\frac{C(\theta)}{|S_1|+|S_2|}}$-close to antipodal. By Lemma \ref{lem:line_proj2}, the set of points $x \in \mathbb{R}^n$ such that $\Pi \q_u^0$ and $x$ are $\epsilon$-close to antipodal is a convex region. Because $\Pi \q_v^0$ lies in this convex region for all $v \in T_u$,  so does the average vector $\Pi \q_{T_u}^0$. Therefore, it follows that $\Pi \q_u^0$ and $\Pi \q_{T_u}^0$ are $\sqrt{\frac{C(\theta)}{|S_1|+|S_2|}}$-close to antipodal. We now show that $\Pi \q_u^0$ and $\Pi \q_{S_2}^0$ are $\left( \frac{1}{\sqrt{12}} - 2 \sqrt{\theta} \right)\frac{1}{\sqrt{|S_1|+|S_2|}}$-close to antipodal. To show this, we use the triangle inequality which gives us
\begin{align*}
    \norm{\alpha \Pi \q_u^0 + (1-\alpha) \Pi \q_{S_2}^0 } &= \norm{\alpha \Pi \q_u^0 + (1-\alpha) \Pi \q_{T_u}^0  + (1-\alpha) \Pi ( \q_{S_2}^0 - q_{T_u}^0) }  \\
    &\le \norm{\alpha \Pi \q_u^0+ (1-\alpha) \Pi \q_{T_u}^0} + |(1-\alpha)| \norm{\q_{T_u}^0- \q_{S_2}^0}.
\end{align*}
We first bound the second term.
$$ |(1-\alpha) |\norm{\p_{S_2}^0-\p_{T_u}^0}^2 =  \frac{\theta}{(1-\theta)|S_2|} \le \frac{2\theta}{|S_2|} \le \frac{4 \theta}{|S_1| + |S_2|}.$$
Then using the fact that $\Pi \q_u^0$ and $\Pi \q_{T_u}^0$ are $\sqrt{\frac{C(\theta)}{|S_1|+|S_2|}}$-close to antipodal, we have 
\begin{align*}
  \min_{0 \le \alpha \le 1}\norm{\alpha \Pi \q_u^0 + \beta \Pi \q_{S_2}^0}^2 &\le \frac{2C(\theta)}{|S_1|+|S_2|} + \frac{8 \theta}{|S_1| + |S_2|}\\
  &=  \left( \frac{1}{\sqrt{12}} - 2 \sqrt{\theta} \right)^2\frac{1}{|S_1|+|S_2|}
\end{align*}
which precisely means that $\Pi \q_u^0$ and $\Pi \q_{S_2}^0$ are $\left( \frac{1}{\sqrt{12}} - 2 \sqrt{\theta} \right)\frac{1}{\sqrt{|S_1|+|S_2|}}$-close to antipodal. Note that $u$ was an arbitrary vertex in $T$. Therefore using the same convexity argument as above, we know that $\Pi \q_{T}^0$ and $\Pi \q_{S_2}^0$ are $\left( \frac{1}{\sqrt{12}} - 2 \sqrt{\theta} \right)\frac{1}{\sqrt{|S_1|+|S_2|}}$-close antipodal.
 However, this is a contradiction to Lemma \ref{lem:large_norm2} so we are done. Hence, there must be at least $\theta^2|S_1||S_2|$ pairs $(u,v)$ where $u \in S_1, v \in S_2$,  such that $\Pi\q_u^0$ and $\Pi \q_v^0$ are $\frac{C(\theta)}{|S_1|+|S_2|}$-far from antipodal. The podal case follows similarly.
\end{proof}

The goal now is to extend Lemma \ref{lem:many_pairs1} to say that that we can find sufficiently many pairs $(u,v)$ such that $\Pi \q_u^0$ and $\Pi \q_v^0$ are far from both antipodal \emph{and} podal (which means that $\Pi \q_u^0$ and $\Pi \q_v^0$ are far from collinear). We do this in Lemma \ref{lem:many_pairs2} but we first present the following supplementary lemma which tells us the conditions under which we can find many pairs of vectors that are far from collinear.

\begin{lemma} \label{lem:supplementary}
Let $\q$ be a vector and let $S$ be a set of vectors such that for all vectors $\mathbf{r} \in S$, $\q$ and $\mathbf{r}$ are $\epsilon$-close to antipodal or $\q$ and $\mathbf{r}$ are $\epsilon$-close to podal. Then for all $\theta \le \frac{1}{100}$, one of the following three cases must occur. Recall that $\q_{S} = \frac{1}{|S|}\sum_{\mathbf{r} \in S} \mathbf{r}$.

\begin{enumerate}
    \item There is a set $S' \subseteq S$ such that for all $\mathbf{r} \in S'$, $\q$ and $\mathbf{r}$ are $\frac{\epsilon}2$-far from collinear and $|S'| = \theta|S|$.
    \item $\q$ and $\q_S$ are $\left(\epsilon + \sqrt{\frac{2 \theta}{|S|}}\right)$-close to antipodal or $\q$ and $\q_S$ are $\left(\epsilon + \sqrt{\frac{2 \theta}{|S|}}\right)$-close to podal (possibly both).
    \item We can find $\theta^4|S|^2$ pairs $(\q_1', \q_2')$ where $\q_1', \q_2' \in S$ such that $\q_1'$ and $\q_2'$ are $\frac{\epsilon}2$-far from collinear.
    \end{enumerate}
\end{lemma}
\begin{proof}
Consider Figure \ref{fig:supplement} along with supplementary Figures \ref{fig:case1}-\ref{fig:case2b}. If there are at least $\theta |S|$ elements of $S$ in the shaded region of Figure \ref{fig:case1}, then we are in case $1$. This is because for every $\mathbf{r}$ in this shaded region of Figure \ref{fig:case1}, both of the line segments from $\q$ to $\mathbf{r}$ and from $-\q$ to $\mathbf{r}$ do not intersect the sphere of radius $\frac{\epsilon}2$ centered at the origin. 

If we are not in case $1$ then we know that greater than $(1-\theta)|S|$ elements of $S$ that lie completely inside the shaded region in Figure \ref{fig:case2a}. We now partition $S$ into three disjoint sets $S_1 \cup S_2 \cup S_3$ where $S_1$ is the set of all $\mathbf{r} \in S$ where $\q$ and $\mathbf{r}$ are $\epsilon$-close to antipodal (but $\epsilon$-far from podal), $S_2$ is the set of all $\mathbf{r} \in S$ where $\q$ and $\mathbf{r}$ are  $\epsilon$-close to podal (but $\epsilon$-far from antipodal), and $S_3$ is the set of all $\mathbf{r}$ in $S$ where $\q$ and $\mathbf{r}$ are $\epsilon$-close to both antipodal and podal.

The geometry implied by the definitions of antipodal and podal means that all of the elements of $S_3$ lie in the shaded region in Figure \ref{fig:case2b}. Therefore, if $|S_3| \ge (1-\theta)|S|$, we can use the triangle inequality as in Lemma \ref{lem:many_pairs1} to bound $\norm{\q_S - \q_{S_3}}$ to show that $\q$ and $\q_S$ are $\left(\epsilon + \sqrt{\frac{2 \theta}{|S|}}\right)$-close to antipodal and podal. Similarly, if either $|S_1| \ge (1-\theta)|S|$ or $|S_2| \ge (1-\theta)|S|$, we can again use the triangle inequality to show that $\q$ and $\q_S$ are $\left(\epsilon + \sqrt{\frac{2 \theta}{|S|}}\right)$-close to antipodal or $\q$ and $\q_S$ are $\left(\epsilon + \sqrt{\frac{2 \theta}{|S|}}\right)$-close to podal which means that we are in case $2$.

Therefore, we can now assume that both $|S_1| \ge \theta |S|$ and $|S_2| \ge \theta |S|$. We now consider the point $\q_{S_1 \cup S_2}$. If this point lines outside the shaded region in Figure \ref{fig:case2a}, then we know that the line segment connecting $\q_{S_1}$ and $\q_{S_2}$ lies outside the shaded region of Figure \ref{fig:case2a} at some point. By our geometric construction, this implies that the entire line segment does not intersect the circle centered at the origin with radius $\frac{\epsilon}2$. Thus, we have that $\q_{S_1}$ and $\q_{S_2}$ are $\frac{\epsilon}2$-far from antipodal. Then by Lemma \ref{lem:many_pairs1}, we know that there are $\theta^2|S_1||S_2|$ pairs $(\q_1', \q_2')$ where $\q_1' \in S_1$ and $\q_2' \in S_2$ such that each pair is also $\frac{\epsilon}2$-far from antipodal. 

We now show that each such pair $(\q_1', \q_2')$ is also $\frac{\epsilon}2$-far from podal. We claim that the line segment connecting $\q_2'$ and $-\q_1'$ cannot intersect the circle of radius $\frac{\epsilon}2$ centered at the origin. This is because all of the points in $S_1$ and $S_2$ have to lie inside the shaded region of Figure \ref{fig:case2a} and these points cannot lie inside the circle of radius $\epsilon$. Thus, the closest the line segment connecting $-\q_{2}'$ and $\q_{1}'$ can come to the circle of radius $\frac{\epsilon}2$ is if $\q_{1}'$ coincides with  the point $E$ and $-\q_{2}'$ coincides with the point $H$ in Figure \ref{fig:supplement}. In this scenario, it is clear from Figure \ref{fig:supplement} that this line segment does not intersect the circle of radius $\frac{\epsilon}2$. Therefore in this case we can find $\theta^2|S_1||S_2|$ many pairs $(\q_1',\q_2')$ that are $\frac{\epsilon}2$-far from antipodal and podal which means we are in case $3$. 

We now consider the case $\q_{S_1 \cup S_2}$ lies inside the shaded region in Figure \ref{fig:case2a}. This implies that the point $\q_S$ also lies inside the shaded region in Figure \ref{fig:case2a}. Therefore, the points $\q$ and $\q_S$ are $\epsilon$-close to antipodal or they are $\epsilon$-close podal. This precisely means that we are again case $2$. \qedhere

\begin{figure}[!htbp]
    \centering
    \includegraphics[scale=0.5]{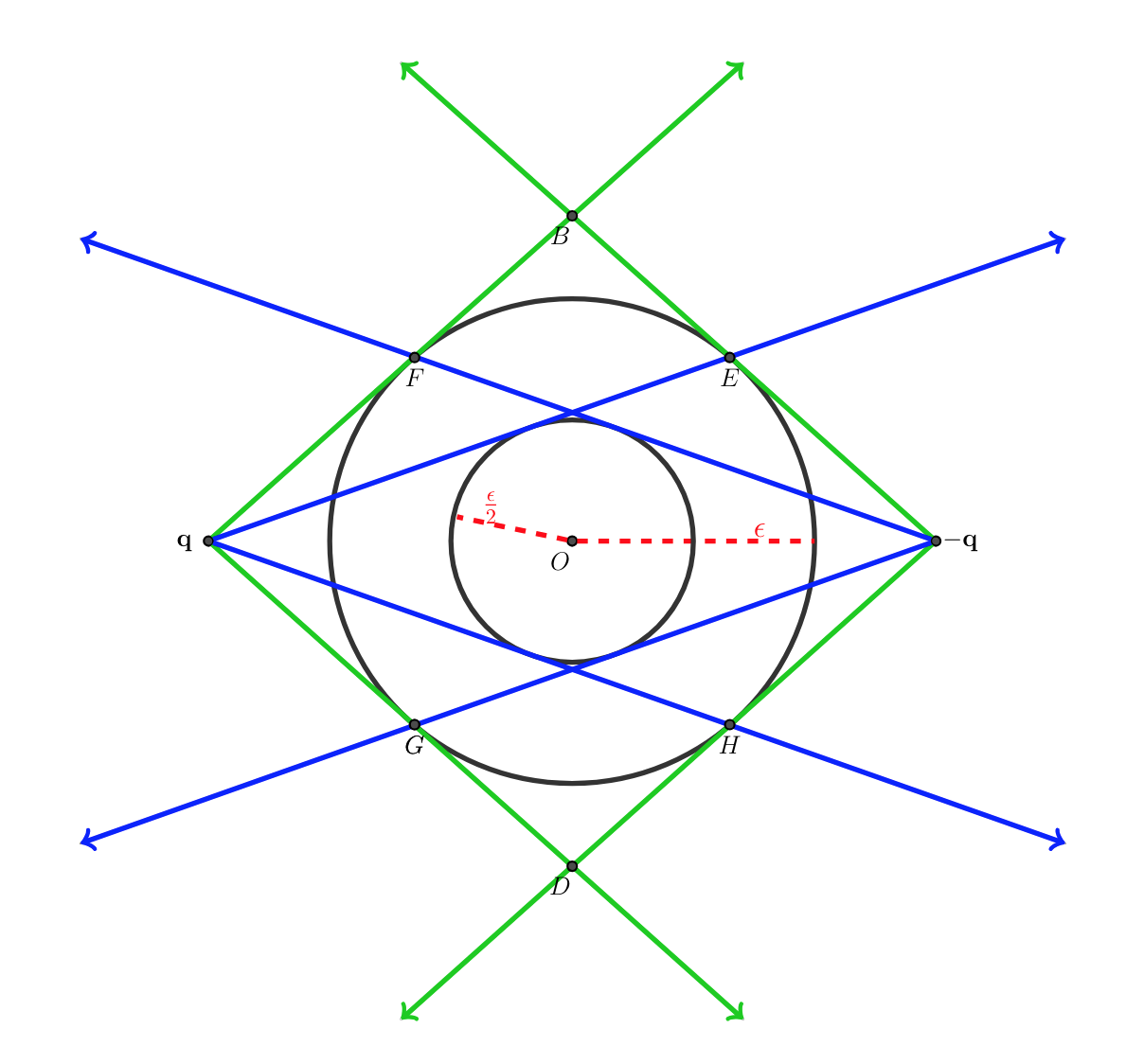}
    \caption{Diagram for Lemma \ref{lem:supplementary}.}
    \label{fig:supplement}
\end{figure}
\begin{figure}[!htbp]
\begin{minipage}{.5\linewidth}
\centering
\subfloat[]{\label{fig:case1}\includegraphics[scale=.35]{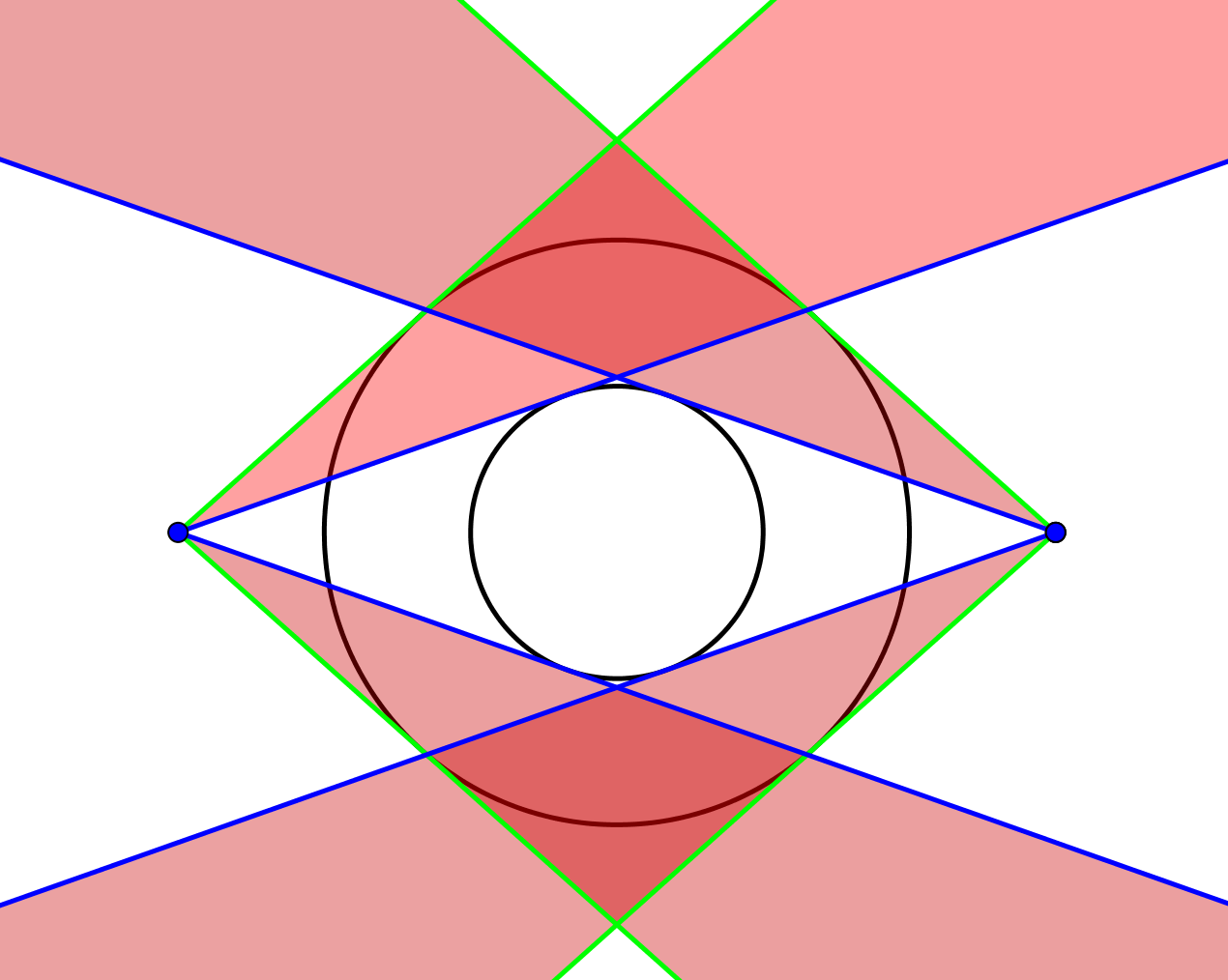}}
\end{minipage}%
\begin{minipage}{.5\linewidth}
\centering
\subfloat[]{\label{fig:case2a}\includegraphics[scale=.35]{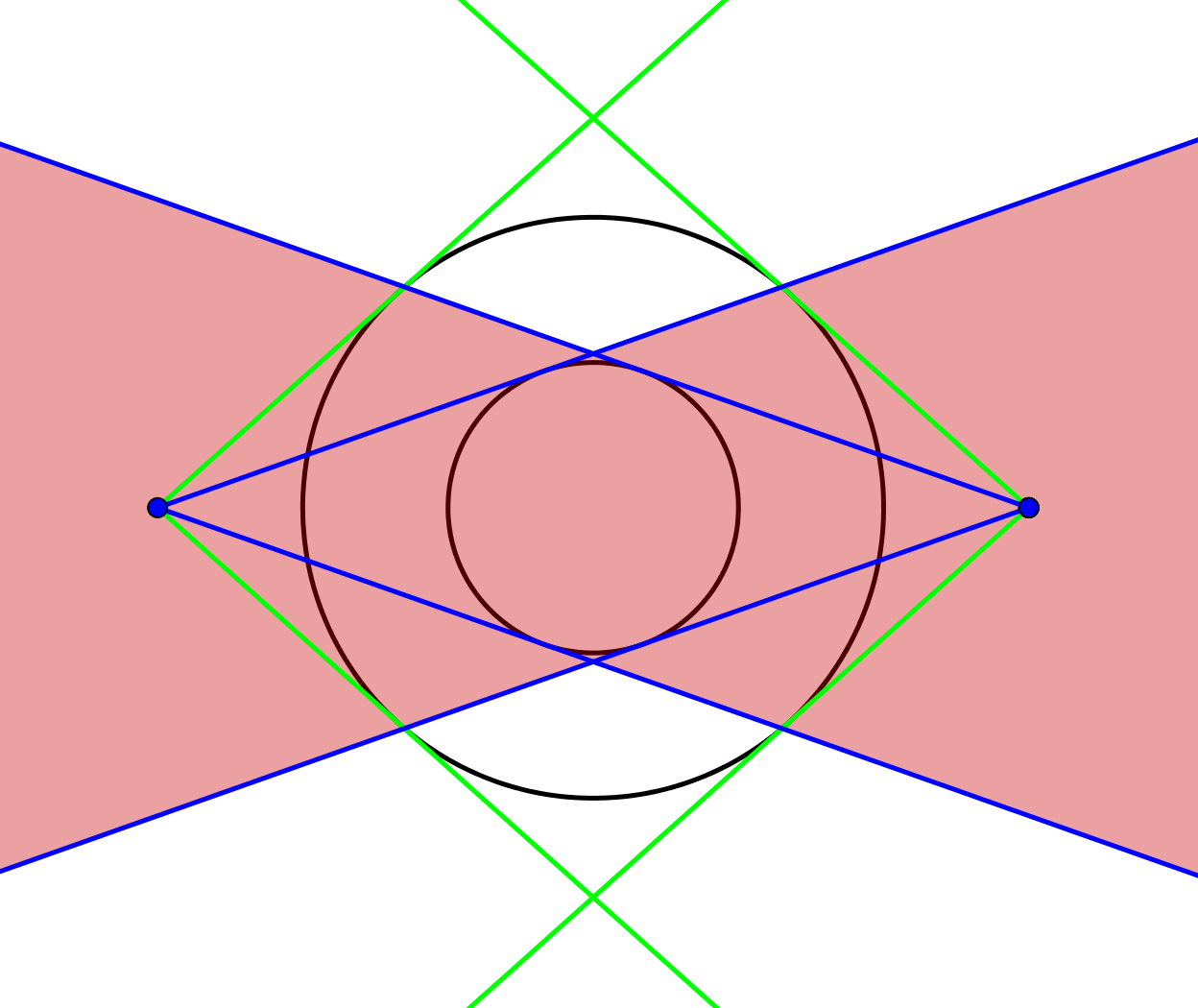}}
\end{minipage}\par\medskip
\centering
\subfloat[]{\label{fig:case2b}\includegraphics[scale=.35]{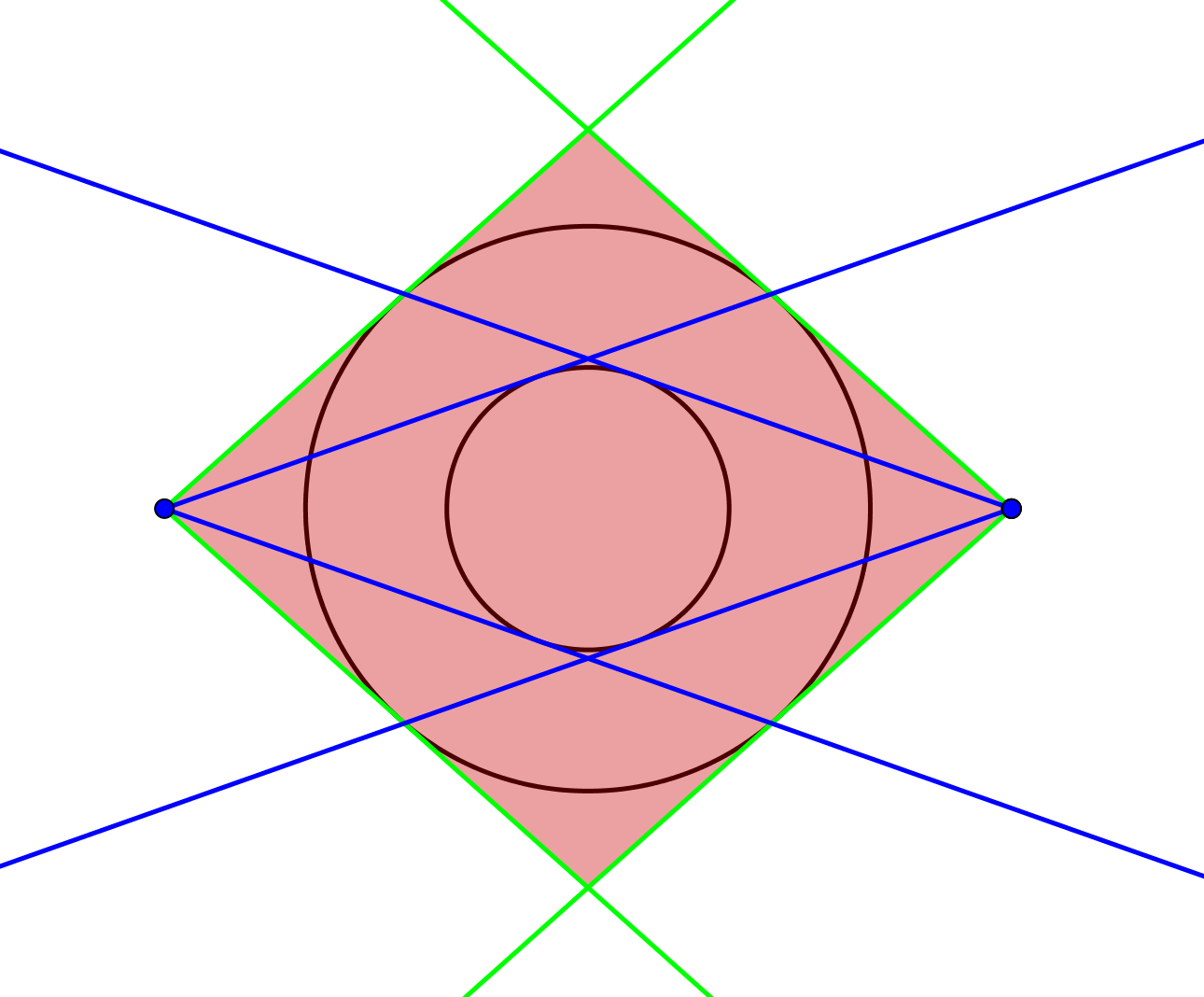}}

\caption{See Figure \ref{fig:supplement} for reference. Various regions corresponding to the cases in Lemma \ref{lem:supplementary} are shaded.}
\label{fig:main}
\end{figure}
\end{proof}
Using Lemma \ref{lem:supplementary} as a stepping stone, we can now extend Lemma \ref{lem:many_pairs1} to prove that there are $\Theta(n^2)$ pairs of vertices $(u,v)$ where $\Pi\q_u^0$ and $\Pi\q_v^0$ are $O\left(\frac{1}{\sqrt{|S_1| + |S_2|}}\right)$-far from collinear.

\begin{lemma} \label{lem:many_pairs2}
Let $S_1$ and $S_2$ be two disjoint subsets of vertices such that the cut $(S_i, V \setminus S_i)$ has conductance less than $\delta$ for $i  \in \{1, 2\}$. Suppose that $|S_1| + |S_2| \le \frac{2n}3$ and let $\Pi$ denote the projection onto the span of the eigenvectors of $\M$ with eigenvalue greater than $1-4\delta$. Let $\theta$ be a sufficiently small constant. Then there are constants $C_1(\theta), C_2(\theta)$ that only depend on $\theta$ such that there are at least $ C_1(\theta) n^2$ pairs of distinct vertices $(u,v)$ such that $\Pi\q_u^0$ and $\Pi\q_v^0$ are $\sqrt{\frac{C_2(\theta)}{|S_1| + |S_2|}}$-far from collinear.
\end{lemma}

\begin{proof}
Let $R(\theta)$ be a constant depending only on $\theta$ that is defined later. If there are at least $\theta^2 |S_1||S_2|$ pairs $(u,v)$ where $u \in S_1$ and $v \in S_2$ such that $\Pi \q_u^0$ and $\Pi \q_v^0$ are $\sqrt{\frac{R(\theta)}{|S_1| + |S_2|}}$-far from collinear then we are done. Otherwise, the number of pairs $(u,v)$ where $u \in S_1, v \in S_2$ and where $\Pi \q_u^0$ and $\Pi \q_v^0$ are $\sqrt{\frac{R(\theta)}{|S_1| + |S_2|}}$-close to antipodal or podal is at least $(1-\theta^2)|S_1||S_2|$. The pigeonhole argument used in Lemma \ref{lem:many_pairs1} implies that there exists $T \subseteq S_1$ of size at least $(1-\theta)|S_1|$ such that for all $u \in T$,  there exists $T_u \subseteq S_2$ of size at least $(1-\theta)|S_2|$ where for all $v \in T_u$,  $\Pi \q_u^0$ and $\Pi \q_v^0$ are $\sqrt{\frac{R(\theta)}{|S_1| + |S_2|}}$-close to either antipodal or podal.

We now fix a particular $u \in T$ and consider the set $T_u$ and let $\q = \Pi \q_u^0$, $S = T_u$, and $\epsilon = \sqrt{\frac{R(\theta)}{|S_1| + |S_2|}}$ in Lemma \ref{lem:supplementary}. If case $3$ in Lemma \ref{lem:supplementary} holds for any $u \in T$ then we are done. Otherwise, if case $1$ holds for at least $\theta|T|$ vertices $u \in T$,  then we are also done. Therefore, we must have that case $2$ holds for at least $(1-\theta)|T| = (1-\theta)^2|S_1|$ vertices $u \in T$. By a similar application of the triangle inequality as in Lemma \ref{lem:large_norm2}, this means that for at least $(1-\theta)^2|S_1|$ vertices $u \in T$,  we have that $\Pi \q_u^0$ and $\Pi \q_{S_2}^0$ are not $\sqrt{\frac{R_1(\theta)}{|S_1| + |S_2|}}$-far from either antipodal or podal for some constant $R_1(\theta)$ that comes from Lemma \ref{lem:supplementary}.

We now consider Lemma \ref{lem:supplementary} again where we take $\q = \Pi \q_{S_2}^0$, the set $S$ to be the set of the $(1-\theta)^2|S_1|$ vertices $u \in T$ described above, and $\epsilon = \sqrt{\frac{R_1(\theta)}{|S_1| + |S_2|}}$. Again if case $3$ holds then we are done. Otherwise, if case $1$ holds for $\theta$-fraction of the vertices $u \in T$ then we are also done. Lastly, if we are in case $2$ for at least $(1-\theta)$-fraction of the vertices $u \in T$, we know that $\Pi \q_{T}^0$ and $\Pi \q_{S_2}^0$ are either $\sqrt{\frac{R_2(\theta)}{|S_1| + |S_2|}}$-close to antipodal or $\sqrt{\frac{R_2(\theta)}{|S_1| + |S_2|}}$-close to podal for some constant $R_2(\theta)$ depending only on $\theta$. This is a contradiction to Lemma \ref{lem:large_norm2} by picking $R(\theta)$ such that $R_2(\theta)= \left( \frac{1}{\sqrt{12}} - 2\sqrt{\theta} \right)^2$ after working through the computations from Lemma \ref{lem:supplementary}.

Therefore in all cases, we can find $C_1(\theta)n^2$ pairs of vertices $(u,v)$ where $\Pi \q_u^0$ and $\Pi \q_v^0$ are $\left( \frac{1}{\sqrt{12}} - 2\sqrt{\theta} \right)^2\frac{1}{|S_1| + |S_2|}$-far from antipodal and podal and hence $\left( \frac{1}{\sqrt{12}} - 2\sqrt{\theta} \right)^2\frac{1}{|S_1| + |S_2|}$-far from collinear. In particular, using Lemma \ref{lem:sparse_cuts}, we can take $C_1(\theta) = (1-\theta)^4\theta^4|S_1|^2$. \qedhere
\end{proof}
Letting $\theta = \frac{1}{100}$ and using the fact that $|S_1| + |S_2| \le \frac{2n}3$, Lemma \ref{lem:many_pairs2} gives us $10^{-20}\epsilon^4n^2$ pairs of vertices $(u,v)$ such that $\Pi \q_u^0$ and $\Pi \q_v^0$ are $\left(\frac{1}{100} \frac{1}{\sqrt{n}} \right)$-far from collinear. 

We now make an observation relating $\Pi \q_u^0, \Pi \q_v^0$ to $\q_u^t, \q_v^t$. Write $\alpha \q_u^0 + \beta \q_v^0 = \sum_i w_i \V_i$ in the eigenbasis of $\M$ and let $H$ denote the set of eigenvalues of $\M$ larger than $1-4 \delta$. Then
$$ \norm{\alpha \Pi \q_u^0 + \beta \Pi \q_v^0}^2 = \sum_{i \in H}w_i^2.$$ Furthermore,
$$ \norm{\alpha \q_u^t + \beta \q_v^t }^2 \ge (1-4\delta)^{2t} \sum_{i \in H} w_i^2 = (1-4\delta)^{2t}\norm{\alpha \Pi \q_u^0 + \beta \Pi \q_v^0}^2.$$
Thus using Lemma \ref{lem:line_proj2}, it follows that if $\Pi \q_u^0$ and $\Pi \q_v^0$ are $\epsilon$-far from antipodal or podal, then $\q_u^t$ and $\q_v^t$ are $(1-4\delta)^t\epsilon$-far from antipodal or podal respectively. Using this observation we translate the result of Lemma \ref{lem:many_pairs2} to the $\q_u^t$ vectors.

\begin{lemma} \label{lem:many_repr_pairs}
Let $G$ be $\epsilon$ far from $(2, \phi^*)$-clusterable. Then there are $10^{-20}\epsilon^4n^2$ pairs of vertices $(u,v)$ such that $\q_u^t$ and $\q_v^t$ are $\left(\frac{1}{100} \, \frac{1}{n^{1/2+ c\mu}} \right)$-far from antipodal and podal where \newline $c = 128c_{3.3}c_{3.10}$ and $c_{3.3}, c_{3.10}$ are constants defined in Lemmas \ref{lem:eigenvalue_gap} and \ref{lem:sparse_cuts} respectively.  
\end{lemma}
\begin{proof}
If $G$ is $\epsilon$-far from $(2, \phi^*)$-clusterable then by Lemma \ref{lem:sparse_cuts}, we can let find sets $S_1$ and $S_2$ such that $|S_1| + |S_2| \le \frac{2n}3$ and the cut $(S_i, V \setminus S_i)$ has conductance less than $\delta$ for each $i \in \{1,2\}$ where $\delta \le c_{3.10} \phi^* \epsilon^{-2} \le c_{3.10} \mu \phi^2. $ Then from Lemma \ref{lem:many_pairs2} and using the fact that $\Pi$ projects into the eigenvectors of $\M$ greater than $1-4\delta$, we know that we can find $10^{-20}\epsilon^4n^2$ pairs of vertices $(u,v)$ such that $\q_u^t$ and $\q_v^t$ are $\left( \frac{1}{100} \, \frac{(1-4\delta)^{t}}{\sqrt{n}} \right)$-far from collinear. Using the fact that $t \ge \frac{32c_{3.3} \log n}{\phi^2}$, we have
$$ \frac{(1-4\delta)^{2t}}{n} \ge \frac{\exp(-4 \delta t)}n \ge \frac{\exp(-128c_{3.3}c_{3.10} \mu \log n)}{n} = \frac{1}{n^{1 + 128c_{3.3}c_{3.10}\mu}},$$
as desired.
\end{proof}
Lemma \ref{lem:many_repr_pairs} says there $\Theta(n^2)$ pairs of vertices $(u,v)$ such that $\q_u^t$ and $\q_v^t$ are $O\left( \frac{1}{n^{1/2 + O(1)\mu}}\right)$-far from collinear. Using this result, we finally show that \textbf{Cluster-Test} with the parameters defined in Theorem \ref{thm:main} also passes the soundness case.
\begin{lemma} \label{lem:algproof_soundness}
\textbf{Cluster-Test} with the parameters defined in Theorem \ref{thm:main} rejects graphs $\epsilon$-far from $(2, \mu \phi^2 \epsilon^2)$-clusterable graphs with probability greater than $ \frac{2}3$.
\end{lemma}
\begin{proof}
Let $G$ be $\epsilon$-far from $(2, \mu \phi^2 \epsilon^2)$-clusterable. We analyze one round of \textbf{Cluster-Test} and lower bound the total rejection probability of one round. Note that \textbf{Cluster-Test} samples a pair of vertices $u$ and $v$ uniformly at random from $G$ at each round. Recall from Lemma \ref{lem:algproof_completeness} that there are three ways one round can reject $G$:
\begin{enumerate}
    \item One of the vertices $u$ or $v$ in the complement of $V'$ in Lemma \ref{lem:small_norm}.
    \item $l_2\textbf{-Norm-Tester}$ rejects $\p_u^t$ or $\p_v^t$ in step $4$ of \textbf{Cluster-Test}.
    \item Both of the eigenvalues of $\tilde{\A}_{u,v}$ are larger than $\Lambda$.
\end{enumerate}
We assume that \textbf{Cluster-Test} does not reject $G$ in cases $1$ and $2$ because these cases can only increase the rejection probability so we focus solely on case $3$.

Let $c =  128c_{3.3}c_{3.10}$ where $c_{3.3}, c_{3.10}$ are constants defined in Lemmas \ref{lem:eigenvalue_gap} and \ref{lem:sparse_cuts} respectively. Suppose that $\q_u^t$ and $\q_v^t$ are $\left(\frac{1}{100} \frac{1}{n^{1/2 + c\mu}}\right)$-far from collinear. Then Lemma \ref{lem:eigsofA} implies that the eigenvalues of $\A_{u,v}$ are larger than $\frac{1}{10^4} \frac{1}{n^{1 + c\mu}}$. Now the matrix $\tilde{\A}_{u,v}$ that $\textbf{Cluster-Test}$ computes can be written as $\tilde{\A}_{u,v} = \A_{u,v} + \mathbf{E}$ where each entry of the $2$ by $2$ matrix $\mathbf{E}$ is at most $\xi$ with probability $1-\eta$ due to $l_2 \textbf{-Inner-Product-Estimator}$. Therefore, $\norm{\mathbf{E}}_F \le 2 \xi$ with probability $(1-\eta)^4$. If this also holds, then by Weyl's inequality
both of the eigenvalues of $\tilde{\A}_{u,v}$ are larger than
$$\frac{1}{10^4} \frac{1}{n^{1 + c\mu}} - 2\xi > \Lambda.$$
Thus, the rejection probability of case $3$ is at least the probability that $\q_u^t$ and $\q_v^t$ are $\left(\frac{1}{100} \frac{1}{n^{1/2 + c\mu}}\right)$-far from collinear times $(1-\eta)^4$.

Now from Lemma \ref{lem:many_repr_pairs}, the probability that $\q_u^t$ and $\q_v^t$ are $\left(\frac{1}{100} \frac{1}{n^{1/2 + c\mu}}\right)$-far from collinear is at least $10^{-20}\epsilon^4$. Therefore, the probability that both the eigenvalues of $\tilde{\A}_{u,v}$ are larger than $\Lambda$ is at least $10^{-20}\epsilon^4 (1-\eta)^4 \ge \frac{10^{-20}\epsilon^4}{16}$. Thus, the probability that $\textbf{Cluster-Test}$ rejects $G$ in one round is at least $ \frac{10^{-20}\epsilon^4}{16}$.
Hence the probability that all $R$ trials of $\textbf{Cluster-Test}$ accept is at most 
$$ \left( 1 -  \frac{10^{-20}\epsilon^4}{16} \right)^R < \exp(-2) < \frac{1}3$$
which means that \textbf{Cluster-Test} rejects $G$ probability greater than $\frac{2}3$. The query complexity is $O(tNR) = O(n^{1/2+O(1)\mu} \cdot \text{poly}(1/\epsilon, 1/\phi,\log n))$. \qedhere
\end{proof}
Lemmas \ref{lem:algproof_completeness} and \ref{lem:algproof_soundness} together prove Theorem \ref{thm:main}, as desired.
\bigbreak
\noindent \textbf{Acknowledgements:}
The authors would like to thank the MIT UROP+ program for the opportunity to work on this project. 

\bibliographystyle{plain}
\bibliography{Paper}
\end{document}